\documentclass[fleqn]{article}

\usepackage{amsmath,textcomp,amssymb,geometry,graphicx,amsthm,mathtools,enumerate,xfrac,pgfplots,
complexity}
\usepackage[mathscr]{euscript}
\usepackage{tikz-cd}

\newtheorem{prop}{Proposition}
\newtheorem{lemma}{Lemma}
\theoremstyle{definition}
\newtheorem{definition}{Definition}[section]

\newcommand{\B}{\mathbb{B}}
\newcommand{\defeq}{\vcentcolon=}
\newcommand{\tran}{^\top}

\newcommand{\lto}{
  \longrightarrow
}

\newcommand{\lmto}{
  \longmapsto
}

\newcommand{\F}{
  \mathbb{F}
}

\renewcommand{\R}{
  \mathbb{R}
}

\newcommand{\N}{
  \mathbb{N}
}
\DeclareMathOperator{\Pref}{Pref}
\DeclareMathOperator{\Suff}{Suff}
\DeclareMathOperator{\ONE}{ONE}

\DeclareMathOperator{\argmin}{argmin}

\newcommand{\one}{\mathbf{1}}

\newcommand{\inv} {
  ^{-1}
}

\DeclareMathOperator{\sgn}{sgn}

\title{Learning DFAs from Confidence Oracles}
\date{May 13, 2020}
\author{Wilson Wu}
\begin{document}
\maketitle
\section{Introduction}
In this report, we discuss the problem of learning a deterministic finite
automaton (DFA) from a \textit{confidence oracle}. That is, we are given access
to an oracle $Q$ with incomplete knowledge of some target language $L$ over an
alphabet $\Sigma$; the oracle maps
a string $x\in\Sigma^*$ to a score in the interval $[-1,1]$ indicating its
confidence that the string is in the language. The interpretation is that the
sign of the score signifies whether $x\in L$, while the magnitude $|Q(x)|$ represents the
oracle's confidence.
Our goal is to learn a DFA representation of the oracle that
preserves the information that it is confident in. That is, the learned DFA
should closely match the oracle wherever it is highly confident, but it need not
do this when the oracle is less sure of itself.

\subsection{Motivation}
\label{sec:motiv}
We will begin by providing a couple examples of confidence oracles that may be
encountered in practice. For one, consider a recurrent neural network (RNN)
that maps finite strings in $\Sigma^*$ to scores in the interval $[-1,1]$ (e.g.\
the last layer is a $\tanh$ nonlinearity). The
RNN outputs can be interpreted in exactly the same way as described for
confidence oracles; i.e., the sign of the score is the classification of whether
the string is in the language, while the magnitude is the RNN's confidence.
Moreover, we can include prior knowledge about the RNN in this confidence. For
example, it may be reasonable to believe that the RNN's accuracy decreases with
the length of the input string (either because the RNN was not trained on long
strings, or due to the problem of vanishing gradients~\cite{pascanu}). If this is the case, we
may simply multiply the RNN's outputs with a factor that penalizes long strings.

For another example, consider a random process (e.g.\ a simulation with random
initializations and/or dynamics) that returns a trace $x$ sampled from some
distribution $X$ and a binary
classification $y\sim Y$ for the trace. We can collect these outputs into a dataset of
independent and identically distributed samples $\{(x_i,y_i)\}_{i\in[N]}$ and
calculate the empirical prior probability $\Pr[Y=1\mid X]$. This probability,
when normalized to the interval $[-1,1]$, can again be treated as a confidence
oracle.

For each of the above examples, the original representation of the confidence
oracle is cumbersome in some way. For instance, modern RNNs may contain millions
of parameters, and it is often difficult for humans to explain or understand
their reasoning. Similar problems arise when considering e.g.\ a simulation with
random outcomes. Therefore, it is desirable to, using only black-box access,
learn a more compact representation for these oracles. DFAs are a natural
choice---they are relatively space/time efficient, and when small enough can be
easily understood by simply examining their transition graphs. Moreover,
model-checking of desired properties is better understood for automata than for
any of the other representations mentioned.

\subsection{Notation}
A DFA $A$ is defined as a tuple
\begin{equation*}
  A\defeq (S,\Sigma,\delta,s_0,F)
\end{equation*}
where $S$ is the set of states, $\Sigma$ is the alphabet,
$\delta:S\times\Sigma\lto S$ is the transition function, $s_0\in S$ is the
initial state, and $F\subseteq S$ are the accepting states. For strings of
length $k$, we recursively
extend the transition function
\begin{equation*}
  \delta^*(s,x)=\delta(\delta^*(s,x_1\ldots x_{k-1}), x_k).
\end{equation*}

For this report, we will conflate a DFA $A$ with its associated language
\begin{equation*}
  L_A\defeq\{x\in\Sigma^*\mid \delta^*(s_0, x)\in F\}.
\end{equation*}
We will also conflate any language
$L\subseteq\Sigma^*$ with the indicator function $\chi_L:\Sigma^*\lto\{-1,1\}$
defined as
\begin{equation*}
  \chi_L(x)\defeq
  \begin{cases}
    1 & x\in L\\
    -1 & x\not\in L.
  \end{cases}
\end{equation*}

\subsection{Problem Statement}
Formally, we are given a finite alphabet $\Sigma$, along with an oracle
$Q:\Sigma^*\lto[-1,1]$. For any other language $A:\Sigma^*\lto\{-1,1\}$, we
define the distance
\begin{equation*}
  d(Q, A)\defeq\sum_{x\in\Sigma^*}|Q(x)|\cdot\mathbf{1}\{\sgn(Q(x))\neq A(x)\}.
\end{equation*}
That is, the distance $d(Q,A)$ is the sum of $Q$'s confidences where the sign of
$Q(x)$ differs from $A(x)$. This definition captures the fact that we are
interested in being close to $Q$ only whenever it is confident.

Let $\mathcal{R}_\Sigma$ be the set of regular languages over $\Sigma$. For a given $\eta>0$, the problem we wish to solve is
\begin{equation}
  \label{eq:prob}
  \begin{aligned}
    & \min_{A\in\mathcal{R}_\Sigma} && |A|\\
    & \text{s.t.} && d(Q,A)\leq \eta,
  \end{aligned}
\end{equation}
where $|\cdot|:\mathcal{R}_\Sigma\lto\N$ maps a regular language to the number
of states in its DFA representation. For practical purposes, it suffices to find
a \textit{bicriteria approximation algorithm} for this optimization problem.
That is, letting $A^*$ be the optimal solution to \eqref{eq:prob}, we hope for
an algorithm that, given $Q$ and $\eta$, returns some automaton $\hat{A}$ such
that $|\hat{A}|\in\poly(|A^*|)$ and $d(Q, \hat{A})\leq c\eta$ for some $\eta$.

For this report, we will also make some assumptions about the oracle $Q$:
\begin{itemize}
\item The sum of $Q$'s confidences converges:
  \begin{equation*}
    \sum_{x\in\Sigma^*}|Q|<\infty.
  \end{equation*}
  This is really the regime we are interested in, for if the sum does not
  converge, any $A$ such $d(Q,A)$ is finite must agree with $Q$ at infinite
  locations. For a concrete example, consider a $Q$ that assigns equal weight to
  any two strings of equal length (that is, the confidence is determined solely
  by the length). Then, if $\sum_{x\in\Sigma^*}|Q|$ diverges, for any $A$ such
  that $d(Q,A)<\infty$, we must necessarily have
  \begin{equation*}
    \sum_{x\in\Sigma^k}\mathbf{1}\{\sgn(Q(x))\neq A(x)\}\lto 0
  \end{equation*}
  as $k$ approaches infinity. That is, the discrepancy between $A$ and $Q$ must
  tend to zero as string
  length increases. This kind of behavior is not very relevant to our practical
  motivations.

  As another consideration, it is likely difficult for an algorithm to learn a
  hypothesis $A$ that agrees with $Q$ on infinite strings using only finite
  queries, especially since no equivalence queries are available (as opposed to
  the setting of Angluin's $L^*$ algorithm).

  Assuming that the confidences converge, we can of course normalize such that
  without loss of generality
  \begin{equation*}
    \sum_{x\in\Sigma^*}|Q|=1.
  \end{equation*}
  That is, the confidences of $Q$ form a probability distribution on $\Sigma^*$.
  Then, we can interpret the distance as
  \begin{equation*}
    d(Q,A)=\Pr[\sgn(Q(x))\neq A(x)]
  \end{equation*}
  when $x$ is sampled from this distribution.
\item
  We assume the distribution defined by $|Q|$ is given in some closed-form expression. In
  particular, it is easy to sample $x\sim|Q|$ and to evaluate the cumulative
  distribution function $\Pr[|x|\leq k]$. Alternatively, we can assume that $Q$
  is enveloped by such a closed-form expression, i.e.
  \begin{equation*}
    \forall x\in\Sigma^*:Q(x)\leq P(x)
  \end{equation*}
  and $P$ has the properties we desire. This is a reasonable assumption; if we
  had only query access to $|Q|$, the oracle could simply hide all of its weight
  on a single string which we could not find without a search through
  exponentially many candidates.
\end{itemize}

\subsection{Interpretation as a metric space}
\label{sec:metric}
Given an oracle $Q$, we can split the information it provides into two parts: a
measure $\pi_Q:2^{\Sigma^*}\lto\R_+$ defined as
\begin{equation*}
  \pi_Q(S)\defeq\sum_{x\in S}|Q(x)|
\end{equation*}
and a language $L_Q$ defined as
\begin{equation*}
  L_Q=\{x\in\Sigma^*\mid Q(x)\geq 0\}.
\end{equation*}
Given our assumption that the confidences converge, we can normalize such
that $\pi_Q(\Sigma^*)=1$; i.e., $\pi_Q$ is just the probability distribution
defined by $|Q|$ as discussed above.

Given the measure $\pi_Q$, it is straightforward to check that the function
$d_Q:2^{\Sigma^*}\times 2^{\Sigma^*}\lto\R_+$ defined as
\begin{equation*}
  d_Q(A,B)\defeq \pi_Q(A\Delta B)
\end{equation*}
is a metric over $\Sigma^*$. This the same as the distance found in the
constraint of~\eqref{eq:prob}; that is, for any automaton $A$,
\begin{equation*}
  d_Q(L_Q, A)=d(Q,A)=\sum_{x\in\Sigma^*}|Q(x)|\cdot\mathbf{1}\{\sgn(Q(x))\neq A(x)\}.
\end{equation*}

Thus, the optimization problem~\eqref{eq:prob} is
that of finding the smallest DFA in the $\eta$-ball centered at some point
$L_Q$ in the space $(2^{\Sigma^*}, d_Q)$. Since for any finite measure $\pi_Q$ we
can always choose a finite set $S\subseteq\Sigma^*$ such that the residual
weight $\pi_Q(\Sigma^*\setminus S)$ is arbitrarily small, the set of DFA is
dense in the metric space. That is, for any $L_Q\in 2^{\Sigma^*}$ and
$\eta>0$, we can choose a finite $S\subseteq\Sigma^*$ such that
\begin{equation*}
 \pi_Q(\Sigma^*\setminus S)\leq\eta.
\end{equation*}
Then, the DFA $A$ deciding the language
$L_Q\cap S$ (which exists since the language is finite) satisfies
\begin{equation*}
  d_Q(L_Q, A)=\pi_Q(L_Q\Delta(L_Q\cap S))=\pi_Q(L_Q\setminus S)\leq \pi_Q(\Sigma^*\setminus S)\leq \eta.
\end{equation*}
In other words, there always exists a feasible solution to the problem~\eqref{eq:prob}.

\subsection{Related Work}
The learning of finite automata has been studied under various settings.
Perhaps the most famous result in this area is Angluin's $L^*$ algorithm
\cite{angluin}, which, given access to membership and equivalence queries on a
target DFA, will learn the automaton exactly. Besides this, most of the results
have been in the negative direction. For example, in the Probably Approximately
Correct setting, we wish to learn a DFA with at most $\varepsilon$ error and
with probability at least $1-\delta$, for any $0<\varepsilon,\delta<1/2$, in
time $\poly(1/\delta,1/\epsilon)$. Kearns and Valiant \cite{kearns} demonstrated
that, under common cryptographic assumptions, the PAC-learning of DFA is hard.
Another hardness result is by Gold \cite{gold}, who
showed that learning the minimal automaton consistent with a finite set of
examples is $\NP$-hard. In fact, as shown by Pitt and Warmuth \cite{pitt}, it is
hard to even approximate the number of states in the minimal automaton to any
polynomial factor.

There is also extensive work on extracting DFAs from various RNNs. Much of this work
uses white-box access to the RNN; that is, the approach is to cluster or
partition the activation space of the RNN in some way, then generate a DFA using
the clusters/partitions as states. See the recent papers by Weiss et
al.~\cite{weiss}  or Wang et al.~\cite{wang} for examples
of this. Usually in this line of work, the quality of the extracted DFA is
evaluated by comparing it to the original RNN on some test set. There are also
few provable guarantees about the resulting DFA's accuracy or size. Our problem
statement differs from this in that there is a rigorously-defined distance that
we wish to minimize, and that we aim to find an algorithm that provably achieves
some distance with (approximately) minimal DFA size. Moreover, we are not
dealing with a white-box RNN; rather, we use a more abstract confidence oracle.

Recently, there has also been much interest in using spectral learning
methods to learn weighted finite automata (WFA). The central technique here is the
``Hankel trick''---the Hankel matrix is an infinite matrix constructed from any mapping
$f:\Sigma^*\lto\R$ with the property that rank factorizations of the matrix
correspond to weighted automata
computing $f$. In practice, it is fine to use a finite sub-block of the Hankel
matrix as long as it is full-rank; the problem is that it is not always clear
how to provide such a full-rank basis. See the survey by Balle et
al.~\cite{balle}
for a survey of spectral learning methods. Our work differs from spectral
learning in that we are learning DFA, which can be thought of as a special case
of WFA. Given a Hankel matrix, it is not clear how to efficiently compute a rank
factorization with the special structure needed for a DFA. Moreover, we do not
assume that we are given a full rank basis.

\section{$\varepsilon$-Closure Algorithm}
\label{sec:eps}
Since $d_Q$ is a metric and thus satisfies the triangle inequality, one
interpretation of the problem is that we are attempting to learn some minimal
DFA $A^*$ given access to an $\eta$-perturbed oracle $L_Q$. That is, since we
define $A^*$ such that $d_Q(A^*,L_Q)\leq \eta$, if we find $A$ such that
$d_Q(A,L_Q)<\eta$, we clearly have by triangle inequality that $d(A,A^*)\leq
\eta$. The intuition for the $\varepsilon$-closure algorithm is that we wish to
use queries of $Q$ to explore each of the states and transitions of the underlying automaton
$A^*$. Once the transition topology is learned, it is not difficult to assign
the correct labels to each state.

More precisely, we use the notion of Brzozowski derivatives: given a language $L\subseteq\Sigma^*$
and string $u\in\Sigma^*$, we define
\begin{equation*}
  u\inv L\defeq \{v\in\Sigma^*\mid u\cdot v\in L\}
\end{equation*}
where $u\cdot v$ signifies the concatenation of $u$ and $v$. The Myhill-Nerode
theorem states that the language $L$ is regular if and only if it has finitely
many distinct Brzozowski derivatives \cite{nerode}; each derivative corresponds
to a state of the DFA. Given an oracle $Q$, we use the metric $d_Q$ on
Brzozowski derivatives of $L_Q$ to define the concept of $\varepsilon$-closure:

\begin{definition}
  Given an alphabet $\Sigma$, a metric $d_Q$ over $\Sigma^*$, a language $L_Q\subseteq\Sigma^*$, and some
  $\varepsilon>0$, we say that a set $S\subseteq\Sigma^*$ is
  \textbf{$\varepsilon$-closed} if
  \begin{equation*}
    \forall u\in S:\exists v\in S: d_Q(u\inv L_Q,v\inv L_Q)\leq \varepsilon.
  \end{equation*}
\end{definition}
Since derivatives correspond to states of a DFA, we think of each string $u\in
S$ as an access string to a state in the true DFA $A^*$.

Our algorithm maintains a set of access strings $S$, initialized to contain only
the empty string $\epsilon$, which accesses the initial state of $A^*$. The idea
is that if $S$ is not $\varepsilon$-closed for a suitable $\varepsilon$, we know
for sure that we have not explored all states of $A^*$. Thus, we iteratively
compute an $\varepsilon$-closure of $S$; the strings in this
$\varepsilon$-closure are the states are the states of the hypothesis DFA, while
the transitions will have been computed along the way. Concretely, the algorithm
is as follows:

\begin{itemize}
\item Initialize the queue $S\defeq\{\epsilon\}$.
\item For each unexplored $u\in S$ and $\sigma\in\Sigma$, let
  \begin{equation*}
    v^*=\argmin_{v\in S}d_Q((u\cdot\sigma)\inv L_Q, v\inv L_Q).
  \end{equation*}
  If $d_Q(u\cdot\sigma,v^*)<\varepsilon$ for some suitable $\varepsilon$, set $\delta(u,\sigma)\defeq v^*$. Otherwise, update
  \begin{equation*}
    S\defeq S\cup\{u\cdot\sigma\}
  \end{equation*}
\item Repeat until each $u\in S$ is explored.
\end{itemize}

Note that since we assume we can efficiently sample from $\pi_Q$, we can easily
approximate the metric $d_Q$ to arbitrary precision with high probability.
Moreover, since $S$ is
implemented as a queue, if $\varepsilon$ is chosen conservatively, the algorithm basically performs a breadth-first search
on a subset of the states of $A^*$ (it is not guaranteed to reach all of them).
Thus, it is guaranteed to terminate in $\poly(|A^*|)$ time.

The problem, however, is that the resulting hypothesis may not be very good. Let
us consider a concrete example where $|Q(x)|=(\lambda \Sigma)^{-|x|}$ for some
factor $\lambda>0$; that is, the confidence of $Q$ decreases geometrically with
the length of $x$. If we choose $\varepsilon$ to be the smallest value that
guarantees we do not add extraneous states, the resulting hypothesis may be
exponentially bad:
\begin{prop}
  Suppose $Q$ is an oracle with $|Q(x)|=(\lambda\Sigma)^{-|x|}$. Then, the
  $\varepsilon$-closure algorithm that sets
  \begin{equation*}
    \varepsilon\defeq 2(\lambda|\Sigma|)^{|v|}
  \end{equation*}
  when considering string $v\in S$
  learns a DFA $\hat{A}$ such that
    $|\hat{A}|\leq |A^*|$ and
    \begin{equation*}
      d(\hat{A},L_Q)\leq (\lambda|\Sigma|)^{O(|\Sigma||A^*|^2)}\eta.
    \end{equation*}
\end{prop}
\begin{proof}
The proof is deferred to Appendix A.
\end{proof}

Intuitively, the problem with this algorithm is that a single access string
provides very limited information about the underlying DFA $A^*$; if we are only
looking at one access string to a state, an adversarial oracle could fool us
into drawing an incorrect transition from this state using little of its
$\eta$ budget. Since there can be many strings accessing this state, the
incorrect transition then badly damages the quality of our hypothesis.

Although
we have considered some potential fixes to this algorithm (e.g.\ altering the
the metric used for $\varepsilon$-closure each iteration or learning more
access strings through sampled counterexamples), none of them seem able to
surmount the core difficulty: single access strings provide little information
about $A^*$, but to acquire more access strings we would first need better
information about $A^*$.

\section{Using SMT and MIP Solvers}
In Section~\ref{sec:eps}, we saw a polynomial-time algorithm that achieved only
exponentially-bad error. Another approach is to express
problem~\eqref{eq:prob} as either an satisfiability modulo theories (SMT)
or a mixed-integer program (MIP);
doing so, we can constrain the solution to attain the bounds we desire. Of
course, by doing this we lose any guarantee that we arrive at the solution in
polynomial time. 

In order to express the distance $d_Q$ in either SMT constraints or MIP
constraints, we need to approximate it with a finite summation. This is not hard
to do, since as we have seen in Section~\ref{sec:metric}, finite languages are
dense in the metric space $(2^{\Sigma^*}, d_Q)$.
More precisely, given
the measure $\pi_Q$ defined by the oracle $Q$, we find some $k\in\N$ such that
\begin{equation}
  \label{eq:pik}
  1-\pi_Q(\Sigma^{\leq k})\leq\eta,
\end{equation}
where $\Sigma^{\leq k}\defeq\bigcup_{j=0}^k \Sigma^j$ is the set of strings up to
length $k$,
and define the truncated measure
\begin{equation*}
  \pi_Q^{\leq k}(x)\defeq
  \begin{cases}
    \pi_Q(x) & |x|\leq k \\
    0 & |x|>k.
  \end{cases}
\end{equation*}
The measure $\pi_Q^{\leq k}$ induces a metric $d_Q^{\leq k}$ that can be computed as a finite
summation:
\begin{equation*}
  d_Q^{\leq k}(L_1, L_2)=\sum_{x\in\Sigma^{\leq k}}|Q(x)|\cdot\one\{L_1(x)\neq L_2(x)\}.
\end{equation*}
Moreover, suppose we find the minimal automaton $A_k$ such that
\begin{equation}
  \label{eq:k}
  d_Q^{\leq k}(L_Q, A_k)\leq \eta.
\end{equation}
Then, by the definition of $d_Q^{\leq k}$, we must have that
\begin{equation*}
  d_Q(L_Q, A_k)\leq 2\eta.
\end{equation*}
Also, since for any $L_Q$ and $A_k$,
\begin{equation*}
  d_Q^{\leq k}(L_Q, A_k)\leq d_Q(L_Q, A_k),
\end{equation*}
the minimal automaton satisfying~\eqref{eq:k} must be no larger than the minimal
automaton satisfying the original constraint in~\eqref{eq:prob}. Notice also
that for the same example considered in Section~\ref{sec:eps} where confidences
decrease geometrically with the length of the input string, for~\eqref{eq:pik}
to be satisfied it suffices for
\begin{equation*}
  k\in O(-\log \eta).
\end{equation*}
That is, for this example, we need only consider $\poly(1/\eta)$ example strings.

Now that we can restrict ourselves to finite example sets without much loss of
generality, we note that the problem we wish to solve is very similar to the one
shown by Pitt and Warmuth~\cite{pitt} to be $\NP$-hard to approximate. The problem
discussed by Pitt and Warmuth is that of finding the minimal automaton
exactly consistent with a finite example set. Given the same examples, what we
wish to do is to find the minimal automaton within some $\eta$-distance.
While the original exact learning problem is $\NP$-hard, there do exist SAT
encodings for the problem, e.g.\ one by Heule and Verwer~\cite{heule}. We can
take any such encoding and transform it into either an SMT formula or MIP
constraint for the $\eta$-distance problem.

\subsection{The Heule-Verwer encoding}
Heule and Verwer's SAT encoding was originally introduced as a translation from
a graph-coloring reduction of the original exact learning problem. We will take
a different approach, instead using the language of weighted finite automata
(WFA). Given an alphabet $\Sigma$, a WFA $A$ of size $n$ over the semiring $\F$ is a tuple
$(\alpha_1,\alpha_\infty, (A_\sigma)_{\sigma\in \Sigma})$ where
\begin{equation*}
  \begin{aligned}
    \alpha_1,\alpha_\infty\in \F^n, && \forall \sigma\in\Sigma: A_\sigma\in \F^{n\times n}.
  \end{aligned}
\end{equation*}
On input $x\in\Sigma^*$ of length $k$, the WFA outputs
\begin{equation*}
  A(x)\defeq \alpha_1\tran A_{x_1}\ldots A_{x_k}\alpha_\infty.
\end{equation*}
For notational convenience, we define $A_x\defeq A_{x_1}\ldots A_{x_k}$.

The key observation here is that DFAs are weighted automata over the Boolean
semiring $\B$, such that the vector (technically, semimodule elements)
$\alpha_1$ and each row of each matrix $A_\sigma$ are all standard
basis vectors. (That is, each vector is $\top$ at exactly one index and $\bot$
everywhere else.) Then, $\alpha_1$ corresponds to the initial state of the DFA,
$\alpha_\infty$ corresponds to the accepting states, and each $A_\sigma$
corresponds to the transitions for the symbol $\sigma$.

This characterization of DFAs suggests a natural SAT encoding for the problem of
exact learning from finite example sets: to determine whether a DFA of size $n$
exists satisfying positive examples $C_+$ and negative examples $C_-$, we simply
check if the conjunction of the formulae
\begin{equation}
  \label{eq:badencode}
  \begin{aligned}
    & \ONE(\alpha_1) \land \forall\sigma\in \Sigma:\forall {j\in[n]}:\ONE(A_{\sigma i})\\
    & \forall x\in C_+: \alpha_1\tran A_x\alpha_\infty=\top\\
    & \forall x\in C_-: \alpha_1\tran A_x\alpha_\infty=\bot
  \end{aligned}
\end{equation}
is satisfiable with any assignment to the variables
$\alpha_1,\alpha_\infty\in\B^n$ and $(A_\sigma)_{\sigma\in\Sigma^*}\in
\B^{n\times n}$.
Here, $A_{\sigma i}$ is the $i$th row of $A_{\sigma}$, and $\ONE$ is the
predicate checking that exactly one element in the vector is $\top$, i.e.
\begin{equation*}
  \ONE(v)\defeq\bigvee_{i=1}^n v_i\land \bigwedge_{i<j}\neg (v_i\land v_j).
\end{equation*}
(When we later move to an MIP encoding in Section~\ref{sec:mip}, we
can replace this with the more natural $\ONE(v)\defeq(\sum_{i=1}^n v_i=1)$.)
The first formula of~\eqref{eq:badencode} states that
$(\alpha_1,\alpha_\infty,(A_\sigma)_{\sigma\in\Sigma})$ indeed forms a DFA,
while the next second and third formulae constrain the DFA outputs to be
consistent with the positive and negative examples, respectively.
Thus, any satisfying assignment corresponds directly to a DFA; if on
the other hand the formula is unsatisfiable, this means there is no DFA of size
$n$ matching the examples.

The problem with this encoding is that it is too large; the multiplication
$A_x=A_{x_1}\ldots A_{x_k}$ results in a number of clauses that is exponential
in the length of the string $x$. The Heule-Verwer encoding gets around this by
considering all example prefixes $\Pref(C_+\cup C_-)$. Each prefix
$x\in\Pref(C_+\cup C_-)$
is assigned a vector $\alpha_x$, and we wish to constrain
\begin{equation*}
  \alpha_x\tran =\alpha_1\tran A_x.
\end{equation*}
This is done by constraining each $\alpha_{x\cdot \sigma}$ with $(x\cdot \sigma)\in P$ to be the
multiplication of $\alpha_{x}$ with the transition matrix $A_\sigma$. That is,
the encoding consists of the formulae
\begin{equation}
  \label{eq:for}
  \begin{aligned}
  & \forall x\in\Pref(C_+\cup C_-):\ONE(\alpha_x)\land \forall\sigma\in \Sigma:\forall {j\in[n]}:\ONE(A_{\sigma i})\\
  & \forall (x\sigma)\in \Pref(C_+\cup C_-): \alpha_{x\sigma}\tran=\alpha_x\tran A_\sigma\\
  & \forall x\in C_+: \alpha_x\tran \alpha_\infty=\top\\
  & \forall x\in C_-: \alpha_x\tran \alpha_\infty=\bot.
  \end{aligned}
\end{equation}
In practice, we can set $\alpha_1$ to be the first standard basis vector $e_1$
by symmetry.
Note that the constraint $\ONE(\alpha_x)$ is actually
redundant for any non-empty $x\neq\epsilon$; however, Heule and Verwer
empirically demonstrated that the addition of these redundant constraints
improves SAT solver performance.

Written from the perspective of WFAs, it is also clear that the choice
of the forward direction
\begin{equation*}
  \alpha_x\tran =\alpha_1\tran A_x.
\end{equation*}
in the Heule-Verwer encoding is more or less arbitrary. It works just as well to
use the backward vectors
\begin{equation*}
  \beta_x = A_x\alpha_\infty.
\end{equation*}
(Notice that as opposed to $\alpha_x$, the vectors $\beta_x$ are not necessarily
standard basis elements.)
Doing so results in the similar encoding
\begin{equation}
  \label{eq:back}
  \begin{aligned}
  & \forall\sigma\in \Sigma:\forall {j\in[n]}:\ONE(A_{\sigma i})\\
  & \forall (\sigma x)\in \Suff(C_+\cup C_-): \beta_{\sigma x}=A_\sigma\beta_x\\
  & \forall x\in C_+: \alpha_1\tran \beta_x=\top\\
  & \forall x\in C_-: \alpha_1\tran \beta_x=\bot.
  \end{aligned}
\end{equation}
Since $\alpha_1$ is constrained to be a standard basis vector while
$\alpha_\infty$ is not, the backwards encoding is somewhat smaller than that in the
forward direction. As seen in Section~\ref{sec:exp}, it is empirically shown to
have moderately better performance. To our knowledge, this variation on the
Heule-Verwer encoding is novel.

\subsection{SMT and MIP encodings for $\eta$-distance}
\label{sec:mip}
In order to solve the $\eta$-distance problem on finite example strings, we
simply replace the exact constraints in~\eqref{eq:back} with new SMT constraints
(over linear real arithmetic) that state the summed deviations from $L_Q$,
weighted by $|Q|$, is no more than $\eta$. That is, we define the function 
$h:\B\lto\R$ that maps $\top\lmto 1$ and $\bot\lmto 0$. Then, we replace the constraints
\begin{equation*}
  \begin{aligned}
    & \forall x\in C_+: \alpha_1\tran \beta_x=\top\\
    & \forall x\in C_-: \alpha_1\tran \beta_x=\bot.
  \end{aligned}
\end{equation*}
with the constraint
\begin{equation}
  \label{eq:smt}
  \sum_{x\in C_+}|Q(x)|(1-h(\alpha_1\tran \beta_x))
  +\sum_{x\in C_-}|Q(x)|(h(\alpha_1\tran \beta_x))\leq \eta.
\end{equation}

While the constraint~\eqref{eq:smt} is specifically for the backwards
encoding~\eqref{eq:back}, the same technique works for the forward encoding, or
in fact any SAT encoding for the problem of exact learning on a finite example set.

Instead of using a SMT encoding, note that we can also treat each of
$\alpha_1,\alpha_\infty,(A_\sigma)_{\sigma\in\Sigma}$ as a $0/1$ vector or
matrix over $\R$. This does not change the DFA definition, since the mapping $h$
is compatible with addition and multiplication as long as we do not try to add
one and one (and this does not happen in the matrix multiplication for DFA).
That is, the above SMT encoding can be translated to an MIP program, where the
variables $\beta_x$ and $A_\sigma$ are constrained to be $0/1$ vectors/matrices
over $\R$.

\subsection{Experiments}
\label{sec:exp}
In order to evaluate the efficiency of our encodings, we measured the time
elapsed for the Z3 SMT solver~\cite{z3} on both the exact learning problem (where the
encoding is a SAT formula) and the $\eta$-distance problem.
The total time includes the time taken to iteratively try the encoding of each increasing
size $n\in\N$ until the smallest DFA is found. (Binary search could also be
used, but since time elapsed increases with the size of the encoding, this is
not too helpful.) In our experiments,
we use the example of oracle confidences that decrease geometrically with input
string length at rate $\lambda=0.9$. Our finite example set is $\Sigma^{\leq k}$, the set of all
strings up to length $k$, for some finite constant $k>0$. All experiments were
run with a single thread on an Intel Core i7-6700 CPU with 16GB of memory.

\begin{figure}[!h]
\caption{Time elapsed vs. size of true DFA for SMT encoding (Z3), $\eta=0$}
\centering
\includegraphics[width=0.6\textwidth]{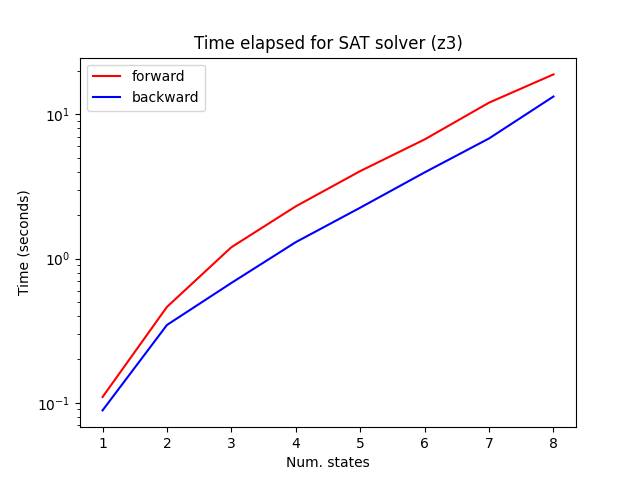}
\label{fig:one}
\end{figure}

\begin{figure}[!h]
\caption{Time elapsed vs. size of true DFA for SMT encoding (Z3), $\eta=10^{-6}$}
\centering
\includegraphics[width=0.6\textwidth]{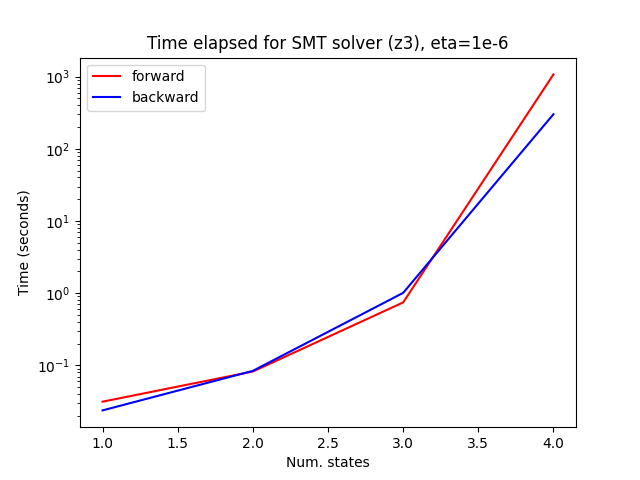}
\label{fig:two}
\end{figure}

See Figure~\ref{fig:one} and Figure~\ref{fig:two} for plots of the time required
by the Z3 solver versus the
number of states in the underlying DFA for the exact learning and
$\eta$-distance problems, respectively. The true language used for each data
point is
\begin{equation}
  \label{eq:mod}
  L_n=\{x\in\Sigma^*\mid \sum_{i=1}^{|x|}\one\{x_i=1\}\equiv 0\mod n\}.
\end{equation}
Note that the minimal automaton for $L_n$ has exactly $n$ states. For this
experiment, the example set $\Sigma^{\leq k}$ was chosen such that $k$ is a
constant larger than the size $n$ of the largest DFA used. From the plots, it is
apparent that, for the exact learning problem, the time elapsed grows
exponentially in the number of states. For the $\eta$-distance problem, the time
seems to grow super-exponentially.

\begin{figure}[!h]
\caption{Time elapsed vs. size of true DFA (Gurobi), $\eta=10^{-6}$}
\centering
\includegraphics[width=0.6\textwidth]{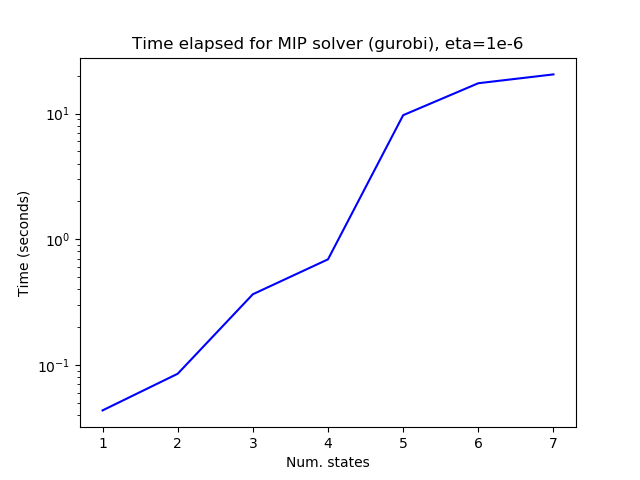}
\label{fig:three}
\end{figure}

We also ran the same time-versus-size experiment for the MIP program, using the
Gurobi solver~\cite{gurobi}; see Figure~\ref{fig:three}. Since the encoding is
formed from the conjunction of several
matrix-multiplication constraints, it is intuitively reasonable that a MIP
optimizer may perform better than an SMT solver. The experiments confirm this
to be the case: as seen in Figure~\ref{fig:two}, the Gurobi solver finds the
minimal automaton in significantly less time. However, the plot also shows that
the time elapsed still grows roughly exponentially with the true number of
states.

\begin{figure}[!h]
\caption{Time elapsed vs. size of SMT encoding (Z3), $\eta=10^{-6}$}
\centering
\includegraphics[width=0.6\textwidth]{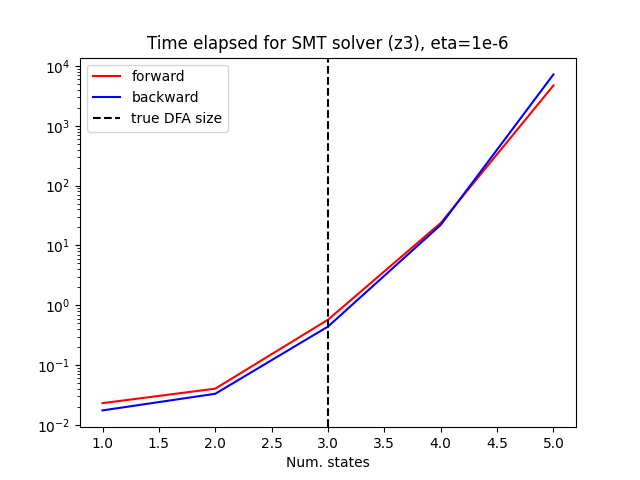}
\label{fig:four}
\end{figure}

\begin{figure}[!h]
\caption{Time elapsed vs. size of SMT encoding (Gurobi), $\eta=10^{-6}$}
\centering
\includegraphics[width=0.6\textwidth]{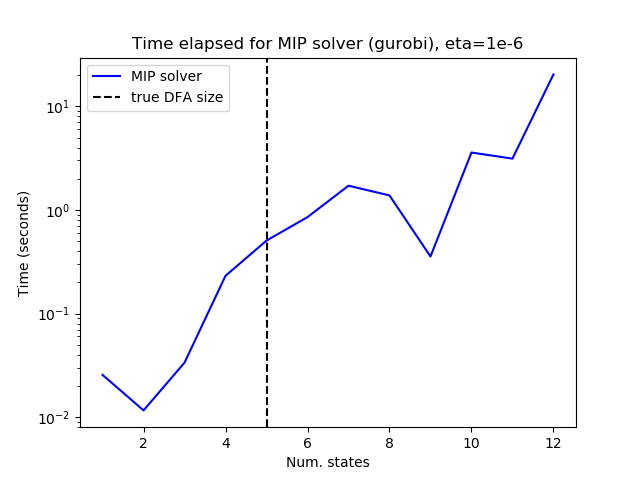}
\label{fig:five}
\end{figure}

Another reasonable hypothesis is that the time required to solve either the SMT
or MIP encoding peaks when the size of the encoding is equal to the size of the
true minimal DFA. This is supported by the intuition that there are many more
DFAs of size greater than minimum, and thus the larger encodings might be
less-constrained, with far more satisfying solutions. However, as our
experiments show (see Figure~\ref{fig:four} and Figure~\ref{fig:five} for the
SMT and MIP encodings, respectively), this intuition does not apply
---the time required by the
solver continues to grow exponentially with the size of the encoding, even
though the size of the true DFA was fixed to a constant size. This is true for
both the SMT and MIP encodings.

\section{Conclusion}

In this report, we covered a polynomial-time algorithm based on the concept of
$\varepsilon$-closure, which turns out to output hypotheses with
exponentially-bad error. We also examined SMT and MIP encodings of the same
problem. For these encodings, the output is guaranteed to be close to the input
oracle since it is constrained to be so. However, there is no guarantee of an
efficient running time, and indeed the time required by both the SMT solver and
MIP optimizer empirically appeared to grow exponentially with the size of the
true minimal DFA. At least qualitatively, this provides some evidence that the
problem of learning minimal DFAs within an $\eta$-distance of confidence oracles
may be computationally hard.

Thus, a prudent next step would be to more carefully examine the complexity of
the optimization problem~\eqref{eq:prob}. Although finding an exact solution
(either deterministically or with high probability) intuitively seems difficult,
we do not yet have any rigorous hardness result. Assuming in the future we do
show this problem is $\NP$-hard, another question we can examine is the hardness
of approximation. That is, we are interested in the bicriteria approximation
problem of finding a DFA that is close in size to the solution
of~\eqref{eq:prob}, while being e.g.\ a constant factor $c\eta$ away from the
input oracle. It is not yet clear what approximation ratios can be achieved for
these two criteria, or whether there is a way of trading off between size and
distance. Note that the $\varepsilon$-closure algorithm already provides one weak
positive result in this direction: it is possible to achieve optimal size with
exponential approximation factor for distance.

Aside from complexity-theoretic concerns, we would also like to solve this
problem in practice, for the reasons listed in Section~\ref{sec:motiv}. In this
vein, we have not yet shown for certain that our SMT and MIP encodings are
impractical. We have only run experiments on the ``mod'' language family
described described in Equation~\ref{eq:mod}; it is possible that these are particularly
difficult languages to learn. Since most languages found in practice are not of
the ``mod'' family, it is worthwhile to experiment on other languages, for
example Tomita grammars or random DFAs.

Furthermore, there might be some way of combining the $\varepsilon$-closure and
SMT approaches to achieve a decent approximation in a practical amount of time.
For example, the SMT solver could be fed information from a partial hypothesis
by the $\varepsilon$-closure algorithm in order to reduce the search space.
Along with this, the SMT solver may be able to provide symbolic counterexamples
that are useful to the $\varepsilon$-closure algorithm. Further investigation is necessary to determine if such an approach is useful in
practice.

Finally, this project touches upon several of the topics covered in EECS 219C.
The most immediate connection is the encoding of our problem in the form of an
SMT formula over linear real arithmetic. The SMT solving techniques discussed in
class are present in some form in the Z3 solver used for this project, along
with many more advanced methods. Moreover, both Angluin's $L^*$ algorithm for
exact learning of DFA, as well as our $\varepsilon$-closure algorithm, can be
considered instantiations of oracle-guided inductive synthesis (OGIS).
Beyond this, the main connection to the course
is within the motivation for this project. That is, the reason we wish to
learn a finite automaton from the confidence oracle in the first place is
because, as seen in class, there exist various model-checking techniques that
work well on automata. Thus, given a confidence oracle in some abstract or
cumbersome representation, it is useful to be able to replace it with a DFA that
functions similarly, but is much easier to formally verify.

\section{Acknowledgements}
Thanks to Marcell Vazquez-Chanlatte and Sanjit Seshia.

\newpage
\bibliography{works}
\bibliographystyle{ieeetr}
\newpage
\appendix
\section{Proof of Proposition 1}
Recall that we are considering the special case where $|Q|=(\lambda
\Sigma)^{-|x|}$. We will first prove some helpful lemmas.

\begin{lemma}
  \label{lem:one}
  For any languages $L_1,L_2\subseteq \Sigma^*$, and any strings
  $u,v\in\Sigma^*$, if $|u|\geq |v|$, then
  \begin{equation*}
    d_Q(u\inv L_2, v\inv L_2)\leq 2(\lambda|\Sigma|)^{|u|}d_Q(L_1,L_2) + d_Q(u\inv L_1+v\inv L_1).
  \end{equation*}

  In particular, if $u\inv L_1=v\inv L_1$ then
  \begin{equation*}
     d_Q(u\inv L_2, v\inv L_2)\leq 2(\lambda|\Sigma|)^{|u|}d_Q(L_1,L_2).
  \end{equation*}
\end{lemma}

\begin{proof}
  First, notice that
  \begin{equation*}
    \begin{aligned}
      d_Q(L_1,L_2)&\geq \sum_{x\in \Sigma^{|u|}}\sum_{y\in \Sigma^*}(\lambda|\Sigma|)^{-|x\cdot y|}\one\{L_1(x\cdot y)\neq L_2(x\cdot y)\}\\
      &=\sum_{x\in \Sigma^{|u|}}(\lambda|\Sigma|)^{-|x|}\sum_{y\in \Sigma^*}(\lambda|\Sigma|)^{-|y|}\one\{L_1(x\cdot y)\neq L_2(x\cdot y)\}\\
      &=(\lambda|\Sigma|)^{-|u|}\sum_{x\in\Sigma^{|u|}}d_Q(x\inv L_1, x\inv L_2)\\
      &\geq (\lambda|\Sigma|)^{-|u|}d_Q(u\inv L_1, u\inv L_2).
    \end{aligned}
  \end{equation*}
  Rearranging, $d_Q(u\inv L_1, u\inv L_2) \leq (\lambda|\Sigma|)^{-|u|}d_Q(L_1,
  L_2)$.

  Using the triangle inequality,
  \begin{equation*}
    \begin{aligned}
      d_Q(u\inv L_2, v\inv L_2)&\leq d_Q(u\inv L_2, u\inv L_1) + d_Q(u\inv L_1, v\inv L_1) + d_Q(v\inv L_1, v\inv L_2)\\
      &\leq 2(\lambda|\Sigma|)^{|u|}d_Q(L_1, L_2)+d_Q(u\inv L_1, v\inv L_1).
    \end{aligned}
  \end{equation*}
  since by assumption $|u|\geq |v|$.
\end{proof}

\begin{lemma}
  \label{lem:two}
  Let $A$ and $A'$ be two DFAs over the same set of states $S$, differing on
  only a single transition. That is, let $\delta$ and $\delta'$ be the
  transition functions for $A$ and $A'$ respectively. Then for some $t,s,s'\in
  S$, and $\sigma\in \Sigma$, we have $\delta(t,\sigma)=s$ while
  $\delta(t,\sigma)=s'$, and otherwise $\delta$ and $\delta'$ agree. Let
  $u,u'\in\Sigma^*$ be access strings to $s,s'$ respectively, and let
  $v\in\Sigma^*$ be the shortest access string to $t$. Then,
  \begin{equation*}
    d_Q(A,A')\leq (\lambda\inv-1)\inv\lambda^{-|v|}d_Q(u\inv A, (u')\inv A).
  \end{equation*}
\end{lemma}

\begin{proof}
  Let $S_t$ be the set of all minimal access strings to $t$, in the sense that
  no prefix of any $x\in S_t$ is also an access string to $t$. Then,
  \begin{equation*}
    \begin{aligned}
    d_Q(A,A')&=\sum_{x\in S_t}\sum_{y\in \Sigma^*}(\lambda|\Sigma|)^{-|x\cdot\sigma\cdot y|}\one\{A(x\cdot \sigma\cdot y)\neq A'(x\cdot \sigma\cdot y)\}\\
    &=\sum_{x\in S_t}\sum_{y\in \Sigma^*}(\lambda|\Sigma|)^{-|x\cdot\sigma\cdot y|}\one\{u\inv A(y)\neq (u')\inv A'(y)\}\\
    &=\sum_{x\in S_t}(\lambda|\Sigma|)^{-(|x|+1)}d_Q(u\inv A, (u')\inv A')\\
    &\leq \sum_{k=|v|}^{\infty}\sum_{x\in\Sigma^k}(\lambda|\Sigma|)^{-(k+1)}d_Q(u\inv A, (u')\inv A')\\
    &\leq\sum_{k=|v|}^{\infty}(\lambda)^{-(k+1)}d_Q(u\inv A, (u')\inv A')\\
    &\leq (\lambda\inv-1)\inv\lambda^{-|v|}d_Q(u\inv A, (u')\inv A').
    \end{aligned}
  \end{equation*}
\end{proof}



\begin{proof}[Proof of Proposition 1]
  Recall that we define $A^*$ to be the minimal automaton such that
  \begin{equation*}
   d_Q(A^*, L_Q)\leq\eta.
  \end{equation*}
  We observe that the algorithm performs a BFS traversal of a subset
  of the states of $A^*$. This is due to the choice of
  $\varepsilon=2(\lambda|\Sigma|)^{|u|}\eta$ when exploring the state accessed
  by $u\in\Sigma^*$. That is, the algorithm only adds a new state $v$ to $S$ if
  \begin{equation*}
    d_Q(u\inv L_Q,v\inv L_Q)\geq 2(\lambda|\Sigma|)^{|u|}\eta\geq 2(\lambda|\Sigma|)^{|u|}d_Q(A^*, L_Q).
  \end{equation*}
  By Lemma~\ref{lem:one}, this can only occur if $d_Q(u\inv A^*, v\inv A^*)>0$;
  i.e., the strings $u$ and $v$ access distinct states in the minimal automaton
  $A^*$. Therefore, the $\varepsilon$-closure algorithm never adds extraneous
  states to $S$; the access strings in the queue represent a subset of the
  states of $A^*$, without duplicates. The fact that this subset is traversed in
  BFS order is simply due to the fact that $S$ implemented as a queue---since
  this is the case, we know that each string in $S$ is in fact the shortest
  access string to a particular state.

  However, it is possible that the transition assigned by the
  $\varepsilon$-closure algorithm is incorrect; that is, while exploring the
  $\sigma$-transition of some state $t$, instead of assigning the correct next
  state $s$, we instead assign some other state $s'$ of $A^*$. Let $v\in\Sigma^*$ be the
  access string in $S$ corresponding to $t$ that we are exploring; by the BFS
  traversal we know that $v$ is the shortest string accessing $t$. Let
  $u,u'\in\Sigma^*$ be access strings to the states $s$ and $s'$, respectively.
  By the definition of the $\varepsilon$-closure algorithm, we know that
  \begin{equation*}
    d_Q(u\inv L_Q, (u')\inv L_Q)\leq\varepsilon=2(\lambda|\Sigma|)^{|v|}\eta.
  \end{equation*}
  By Lemma~\ref{lem:one}, this can only occur if
  \begin{equation*}
    d_Q(u\inv A^*, (u')\inv A*)\leq 2(\lambda|\Sigma|)^{|u|}d_Q(A^*, L_Q)+\varepsilon\leq 4(\lambda|\Sigma|)^{|u|}\eta.
  \end{equation*}

  The idea now is that whenever the algorithm mistakenly assigns a transition to
  $s$ instead of $s'$, we know that the two states are not too far apart in the
  true automaton $A^*$. Since our hypothesis automaton $A$ will have the same set of
  states as $A^*$, with some number of incorrect transitions, we can use
  Lemma~\ref{lem:two} to bound the distance from $A$ to $A^*$ (and thus
  $d_Q(A,L_Q)$, by triangle inequality). More precisely, let us think of the $A$
  as the result of a sequence of no more than $|\Sigma|n$ mistakes
  (since there are only this
  many transitions in $A^*$); this corresponds to a sequence of automata
  $A^{(1)},\ldots,A^{(k)}$ on the same states where $A^{(1)}=A$ is the true automaton and
  $A^{(k)}=A$ is our hypothesis, and each consecutive pair $A^{(j)},A^{(j+1)}$
  differs on exactly one transition. That is, letting $\delta^{(j)}$ and
  $\delta^{(j+1)}$ be the transition functions for $A^{(j)}$ and $A^{(j+1)}$
  respectively, we have for some states $t_j, s_j, s'_j$ and $\sigma_j\in\Sigma$
  that
  \begin{equation*}
    \delta^{(j)}(t_j,\sigma_j)=s_j
  \end{equation*}
  while
  \begin{equation*}
    \delta^{(j+1)}(t_j,\sigma_j)=s'_j.
  \end{equation*}
  By Lemma~\ref{lem:one},
  \begin{equation*}
    \begin{aligned}
      d_Q(s_{j+1}\inv A^{(j+1)}, (s'_{j+1})\inv A^{(j+1)})&\leq 2(\lambda|\Sigma|)^nd_Q(A^{(j)}, A^{(j+1)})+d_Q(s_{j+1}\inv A^*, (s'_{j+1})\inv A^*)\\
      &\leq 2(\lambda|\Sigma|)^nd_Q(A^{(j)}, A^{(j+1)})+4(\lambda|\Sigma|)^{n}\eta.
    \end{aligned}
  \end{equation*}
  By Lemma~\ref{lem:two},
  \begin{equation*}
    d_Q(A^{(j)}, A^{(j+1)})\leq (\lambda\inv -1)\inv d_Q(s_j\inv A^{(j)}, (s'_j)\inv A^{(j)})
  \end{equation*}
  (we drop the $\lambda^{|s_j|}$ term, which is less than one.)
  Thus, we have the recurrence
  \begin{equation*}
    d_Q(A^{(j)}, A^{(j+1)})\leq 2(\lambda\inv-1)(\lambda|\Sigma|)^{n}d_Q(A^{(j-1)}, A^{(j)})+4(\lambda|\Sigma|)^n.
  \end{equation*}
  Unraveling this recurrence using the triangle inequality and the fact that $k\leq |\Sigma|n$, we have
  \begin{equation*}
    d_Q(A^*,A)=d_Q(A^{(1)},A^{(k)})\leq (\lambda|\Sigma|)^{O(|\Sigma|n^2)}\eta
  \end{equation*}
  as desired.

  Note that the $\varepsilon$-closure algorithm provides only the transition
  topology of $A$, but once this is known it is not difficult to learn the
  accepting states with only $O(\eta)$ additional error.
\end{proof}
Admittedly, the bounds in the proof are fairly crude. However, even with a more
refined analysis, the exponential term remains.
\end{document}